\documentclass[journal,onecolumn,12pt]{IEEEtran}
\usepackage{pslatex} 
\usepackage{amsfonts,color,morefloats,times}
\usepackage{amssymb,amsmath,latexsym,amsthm}

\newcommand{\diff}{{\mathrm{diff}}} 
\newcommand{\wt}{{\mathrm{wt}}} 
\newcommand{\image}{{\mathrm{Im}}}

\newcommand{\Z}{\mathbb{{Z}}}

\newcommand{\tr}{{\mathrm{Tr}}}
\newcommand{\gf}{{\mathrm{GF}}}

\newcommand{\PG}{{\mathrm{PG}}}

\newcommand{\C}{{\mathcal{C}}}
\newcommand{\rC}{{\mathrm{c}}} 

\newcommand{\cP}{{\mathcal{P}}}
\newcommand{\cB}{{\mathcal{B}}} 
\newcommand{\cI}{{\mathcal{I}}} 
\newcommand{\cD}{{\mathcal{D}}} 
\newcommand{\bD}{{\mathcal{D}}} 

\newcommand{\bc}{{\mathbf{c}}}

\newtheorem{theorem}{Theorem}
\newtheorem{lemma}[theorem]{Lemma}

\newtheorem{corollary}[theorem]{Corollary}
\newtheorem{conj}{Conjecture}

\newtheorem{example}{Example}

\begin{document}

\title{Linear Codes from Some $2$-Designs\thanks{The research of C. Ding was supported by The Research Grants Council of Hong Kong, under Grant No. 16301114.}}

\author{Cunsheng Ding\thanks{C. Ding is with the Department of Computer Science and Engineering, 
The Hong Kong University of Science and Technology, Clear Water Bay, Kowloon, Hong Kong. Email: cding@ust.hk}}

\date{\today}
\maketitle

\begin{abstract}
A classical method of constructing a linear code over $\gf(q)$ with a $t$-design is to use the incidence 
matrix of the $t$-design as a generator matrix over $\gf(q)$ of the code. This approach has been extensively investigated in the 
literature. In this paper, a different method of constructing linear codes using specific classes of 
$2$-designs is studied, and linear codes with a few weights are obtained from almost difference sets, 
difference sets, and a type of $2$-designs associated to semibent functions. Two families of the codes 
obtained in this paper are optimal. 
The linear codes presented in this paper have applications in secret sharing and authentication schemes, 
in addition to their applications in consumer electronics, communication and data storage systems.   
 A coding-theory approach to the characterisation of highly nonlinear 
Boolean functions is presented.       
\end{abstract}


\section{Introduction}\label{sec-intro} 

Throughout this paper, let $p$ be an odd prime and let $q=p^m$ for some positive integer $m$. 
An $[n,\,\kappa,\,d]$ code $\C$ over $\gf(p)$ is a $\kappa$-dimensional subspace of $\gf(p)^n$ with minimum 
(Hamming) distance $d$. 
Let $A_i$ denote the number of codewords with Hamming weight $i$ in a code
$\C$ of length $n$. The {\em weight enumerator} of $\C$ is defined by
$
1+A_1z+A_2z^2+ \cdots + A_nz^n.
$ 
A code $\C$ is said to be a $t$-weight code  if the number of nonzero
$A_i$ in the sequence $(A_1, A_2, \cdots, A_n)$ is equal to $t$.

A \emph{finite incidence structure}, denoted by $(\cP, \,\cB, \,\cI)$, consists of two disjoint finite sets 
$\cP$ and $\cB$, and a subset $\cI$ of $\cP \times \cB$. The members of $\cP$ are called \emph{points} 
and are normally denoted by lower-case Roman letters; the members of $\cB$ are referred to as \emph{blocks} 
and are normally denoted by upper-case Roman letters. If the ordered pair $(p, \,B)$ is in $\cP \times \cB$, we 
say that $p$ is \emph{incident} with $B$, or $B$ contains the point $p$, or that $p$ is on $B$. 

An incidence structure $\cD=(\cP, \,\cB, \,\cI)$ is called a \emph{$t$-$(n, \,k, \,\lambda)$ design}, or simply a 
\emph{$t$-design}, where $t$, $n$, $k$ and $\lambda$ are nonnegative integers, if 
\begin{enumerate}
\item $|\cP|=n$; 
\item every block $B \in \cB$ is incident with precisely $k$ points; 
\item every $t$ distinct points are together incident with precisely $\lambda$ blocks. 
\end{enumerate}  

Let $\cD = (\cP, \,\cB, \,\cI)$ be an incidence structure with $n \ge 1$ points and $b \ge 1$ blocks. 
The points of $\cP$ are usually indexed with $p_1,p_2,\cdots,p_n$, and the blocks of $\cB$ 
are normally denoted by $B_1, B_2, \cdots, B_b$. The {\em incidence matrix\index{incidence 
matrix}} $M_\cD=(m_{ij})$ of $\cD$ is a $b \times n$ matrix where $m_{ij}=1$ if  $p_j$ is on $B_i$ 
and $m_{ij}=0$ otherwise. It is clear that the incidence matrix $M_\cD$ depends on the labeling of 
the points and blocks of $\cD$, but is unique up to row and column permutations. Conversely, 
every ($0, \, 1$)-matrix (entries are $0$ or $1$) determines an incidence structure. Our definition 
of the incidence matrix follows reference \cite{AssmusKey92}[p.12]. In some other references, the transpose of the 
matrix $M_\cD$ above is defined as the incidence matrix.  When $M_\cD$ is viewed as a matrix 
over $\gf(p)$, it spans a linear code of length $n$ over $\gf(p)$, denoted by $\C_p(\cD)$ and  
called the linear code of the incidence structure. With this framework of construction, every 
$t$-design yields a linear code over $\gf(p)$. This approach to the construction of linear codes 
with $t$-designs has been extensively studied.   

The objective of this paper is to study a different method of constructing linear codes using certain 
special types of $2$-designs. In this paper, one-weight, two-weight and three-weight linear 
codes are obtained from almost difference sets, difference sets and a type of 2-designs associated 
with semibent functions. Two families of the linear codes presented in this paper are optimal. 
The linear codes with a few weights presented in this paper have applications in secret sharing 
\cite{ADHK} and authentication codes \cite{CX05}, 
in addition to their applications in consumer electronics, communication and data storage systems.    
 A coding-theory approach to the characterisation of highly nonlinear 
Boolean functions is presented in this paper. 

\section{Mathematical foundations} 

\subsection{Almost difference sets and difference sets} 

For convenience later, we define the \emph{difference function\index{difference function}} of a subset 
$D$ of $(A,\,+)$ as 
\begin{eqnarray}\label{eqn-DifferenceFunction}
\diff_D(x)=|D \cap (D+x)|, 
\end{eqnarray}
where $D+x=\{y+x: y \in D\}$. 

A subset $D$ of size $k$ in an abelian group $(A, \, +)$ with order $v$ is called 
a $(v, \,k, \,\lambda)$ \emph{difference set\index{difference set}}  in $(A,\,+)$ if the difference function 
$\diff_D(x)=\lambda$ for every nonzero $x \in A$. A difference set $D$ in $(A,\,+)$ is called \emph{cyclic} if 
the abelian group $A$ is so. 


If $D$ is a $(v,\,k, \,\lambda)$ difference set in $(A,\,+)$, its \emph{complement\index{complement of a difference set}}, 
$D^{\rC}=A \setminus D$, is a $(v, \,v-k, \,v-2k+\lambda)$ difference set in  $(A,\,+)$. 


Let $D$ be a $(v, \,k, \,\lambda)$ difference set in an abelian group $(A, \,+)$. We associate $D$ with 
an incidence structure $\bD$, called the \emph{development\index{development of difference sets}} 
of $D$, by defining $\bD=(\cP, \,\cB, \,\cI)$, where $\cP$ is the set of the elements in $A$, 
$$ 
\cB=\{a+D: a \in A\}, 
$$
and the incidence $\cI$ is the membership of sets. Each block $a+D=\{a+x: x \in D\}$ is called a  
\emph{translate\index{translate}} of $D$. The development $\bD$ of a difference set $D$ is called 
a \emph{difference set design\index{difference set design}} 
and also the \emph{translate design\index{translate design}} of $D$.

Let $\bD$ be the development of a $(v, \,k, \,\lambda)$ difference set $D$ in a group $A$. Then $\bD$ 
is a $2$-$(v, \, k, \, \lambda)$ design \cite[Theorem 4.4.1]{AssmusKey92}. Every difference set design 
defines a linear code $\C_p(\bD)$ automatically, which was introduced in Section \ref{sec-intro}.

Let $(A, \,+)$ be an abelian group of order $v$. A $k$-subset $D$ of  $A$ is a $(v, \,k, \,\lambda, \,t)$ 
\emph{almost difference set\index{almost difference set}} of $A$ if the difference function $\diff_D(x)$ 
takes on $\lambda$ altogether $t$ times and  $\lambda +1$ altogether $v-1-t$ times when $x$ ranges 
over all the nonzero elements of $A$. In the sequel, we will employ some almost difference sets to construct linear codes with only a few weights. 

\subsection{Group characters in $\gf(q)$}

An {\em additive character} of $\gf(q)$ is a nonzero function $\chi$ 
from $\gf(q)$ to the set of nonzero complex numbers such that 
$\chi(x+y)=\chi(x) \chi(y)$ for any pair $(x, y) \in \gf(q)^2$. 
For each $b\in \gf(q)$, the function
\begin{eqnarray}\label{dfn-add}
\chi_b(c)=\epsilon_p^{\tr(bc)} \ \ \mbox{ for all }
c\in\gf(q) 
\end{eqnarray}
defines an additive character of $\gf(q)$, where and whereafter $\epsilon_p=e^{2\pi \sqrt{-1}/p}$ is 
a primitive complex $p$th root of unity. When $b=0$,
$\chi_0(c)=1 \mbox{ for all } c\in\gf(q), 
$ 
and is called the {\em trivial additive character} of
$\gf(q)$. The character $\chi_1$ in (\ref{dfn-add}) is called the
{\em canonical additive character} of $\gf(q)$. 
It is known that every additive character of $\gf(q)$ can be 
written as $\chi_b(x)=\chi_1(bx)$ \cite[Theorem 5.7]{LN}.

\section{A generic construction of linear codes}

\subsection{The description of the linear codes}

Let $D=\{d_1, \,d_2, \,\ldots, \,d_n\} \subseteq \gf(q)$, where again $q=p^m$.
Let $\tr$ denote the trace function from $\gf(q)$ onto $\gf(p)$ throughout 
this paper. We define a linear code of 
length $n$ over $\gf(p)$ by 
\begin{eqnarray}\label{eqn-maincode} 
\C_{D}=\{(\tr(xd_1), \tr(xd_2), \ldots, \tr(xd_n)): x \in \gf(q)\},   
\end{eqnarray}  
and call $D$ the \emph{defining set} of this code $\C_{D}$. By definition, the 
dimension of the code $\C_D$ is at most $m$.

This construction is generic in the sense that many classes of known codes 
could be produced by selecting the defining set $D \subseteq \gf(q)$. This 
construction technique was employed in \cite{DN07} and \cite{DLN} for 
obtaining linear codes with a few weights. The objective of this paper is 
to construct linear codes $\C_D$ using almost difference sets $D$ and 
difference sets $D$ as well as a few classes of $2$-designs defined later. If the 
set $D$ is well chosen, the code $\C_D$ may have good or optimal parameters. 
Otherwise, the code $\C_D$ could have bad parameters.

\subsection{The weights in the linear codes $\C_D$}

It is convenient to define for each $x \in \gf(q)$, 
\begin{eqnarray}\label{eqn-mcodeword}
\bc_{x}=(\tr(xd_1), \,\tr(xd_2), \,\ldots, \,\tr(xd_n)).  
\end{eqnarray} 
The Hamming weight $\wt(\bc_x)$ of $\bc_x$ is $n-N_x(0)$, where  
$$ 
N_x(0)=\left|\{1 \le i \le n: \tr(xd_i)=0\}\right| 
$$ 
for each $x \in \gf(q)$. 

It is easily seen that for any $D=\{d_1,\,d_2,\,\ldots, \,d_n\} \subseteq \gf(q)$
we have 
\begin{eqnarray}\label{eqn-hn3}  
pN_x(0) 
= \sum_{i=1}^n \sum_{y \in \gf(p)} e^{2\pi \sqrt{-1} y\tr(xd_i)/p} \nonumber 
= \sum_{i=1}^n \sum_{y \in \gf(p)} \chi_1(yxd_i) \nonumber 
= n + \sum_{y \in \gf(p)^*} \chi_1(yxD) 
\end{eqnarray} 
where $\chi_1$ is the canonical additive character of $\gf(q)$, $aD$ denotes the set 
$\{ad: d \in D\}$, and $\chi_1(S):=\sum_{x \in S} \chi_1(x)$ for any subset $S$ of $\gf(q)$.  
Hence, 
\begin{eqnarray}\label{eqn-weight}
\wt(\bc_x)=n-N_x(0)=\frac{(p-1)n-\sum_{y \in \gf(p)^*} \chi_1(yxD)}{p}. 
\end{eqnarray}

\section{Linear codes from skew sets} 

A subset $D$ of $\gf(q)^*$ is called a skew set of $\gf(q)$ if $D$, $-D$ and $\{0\}$ form a partition of 
$\gf(q)$. 

\begin{theorem}\label{thm-part2}
Let $D$ be any skew set of $\gf(q)$. 
Then $\C_D$ is a one-weight code over $\gf(p)$ with parameters 
$[(q-1)/2, \,m, \,(p-1)q/2p]$. 
\end{theorem} 

\begin{proof}
Let $D$ be a skew set of $\gf(q)$. By definition $|D|=(q-1)/2$ and $xD$ is also a skew set of $\gf(q)$ for 
every $x \in \gf(q)^*$. For any $x \in \gf(q)^*$, it follows from (\ref{eqn-weight}) that 
\begin{eqnarray*}
\wt(\bc_x) = \frac{(p-1)(q-1)/2-\sum_{y=1}^{(p-1)/2} \chi_1(y xD \cup (-y xD))}{p} 
                   = \frac{(p-1)q}{2p}. 
\end{eqnarray*}
The desired conclusions then follow. 
\end{proof}

We remark that the code of Theorem \ref{thm-part2} is optimal as it meets the Griesmer bound. 

A skew Hadamard difference set in $(\gf(q),\,+)$ is both a difference set and a skew set. By definition, any 
skew Hadamard difference set in $(\gf(q),\,+)$ must have parameters $(q, \,(q-1)/2, \,(q-3)/4)$, where $q 
\equiv 3 \pmod{4}$. According to Theorem \ref{thm-part2}, any skew Hadamard difference set in $(\gf(q),\,+)$ 
gives a one-weight code over $\gf(p)$. 

The first skew Hadamard difference set was the Paley set $D$, which is set of all nonzero squares in $\gf(q)$, 
where $q \equiv 3 \pmod{4}$. 
Recently, many new constructions of skew Hadamard difference sets have been discovered. For detailed 
information, the reader is referred to \cite{DingYuan06}, \cite{DWX07}, \cite{FengXiang}, \cite{Momihara}, 
and \cite{WQWX07}.

\section{Linear codes from the images of functions $f$ on $\gf(q)$} 

Let $f(x)$ be a function from $\gf(q)$ to $\gf(q)$. We define 
$$ 
D(f):=\{f(x): x \in \gf(q)\} \setminus \{0\}.  
$$
In this section, we consider the code $\C_{D(f)}$. In general, it is difficult to determine the length $n_f:=|D(f)|$ of this code, 
not to mention its weight distribution. However, in certain special cases, the parameters and the weight distribution of $\C_{D(f)}$ 
can be settled.

\subsection{The codes $\C_{D(f)}$ from quadratic functions over $\gf(p^m)$ for odd $p$}  

Throughout this section, let $q=p^m$ be odd. 
A polynomial $f$ over $\gf(q)$ of the form 
$$ 
f(x)=\sum_{i \in I} \sum_{j \in J} a_{i,j} x^{p^i+p^j}  
$$ 
is called a {\em quadratic form} over $\gf(q)$, where $a_{i,j} \in \gf(q)$,  and 
$I$ and $J$ are subsets of $\{0,1,2, \ldots, m-1\}$. 

Note that $\gf(q)$ is a vector space of dimension $m$ over $\gf(p)$. The rank of the quadratic form $f$ over $\gf(q)$ 
is defined to be the codimension of the $\gf(p)$-vector space 
$$ 
V_f=\{x \in \gf(q): f(x+z)-f(x)-f(z)=0 \mbox{ for all } z \in \gf(q)\}.  
$$
That is $|V_f|=p^{m-r}$, where $r$ denotes the rank of $f$. 

It is still very difficult to determine the length $n_f$ of the code $\C_{D(f)}$ for general quadratic forms $f$, 
let alone the weight distribution of the code $\C_{D(f)}$. However, under certain conditions the 
weight distribution of $\C_{D(f)}$ can be worked out. To this end, we need the following lemma \cite{ZD13}. 

\begin{lemma}\label{lem-ZD13} 
Let $f$ be a quadratic form of rank $r$ over $\gf(q)$. If $r$ is even, then 
$$ 
\sum_{y \in \gf(p)^*} \sum_{x \in \gf(q)} \chi_1(yf(x)) = \pm (p-1)p^{m-\frac{r}{2}}. 
$$
If $r$ is odd, then 
$$ 
\sum_{y \in \gf(p)^*} \sum_{x \in \gf(q)} \chi_1(yf(x)) = 0. 
$$
\end{lemma} 

We are now ready to prove the following theorem. 

\begin{theorem}\label{thm-qfcodes}
Let $f$ be a quadratic form of rank $r$ over $\gf(q)$ such that 
\begin{itemize}
\item $f(0)=0$ and $f(x) \neq 0$ for all $x \in \gf(q)^*$; and 
\item $f$ is $e$-to-$1$ on $\gf(q)^*$ (i.e. $f(x)=u$ has either $e$ solutions $x \in \gf(q)^*$ or no solution for each $u \in \gf(q)^*$), 
          where $e$ is a positive integer. 
\end{itemize}

If $r$ is odd, then $\C_{D(f)}$ is a one-weight code over $\gf(p)$ with parameters 
$[(q-1)/e, \,m, \,(p-1)q/ep]$. 

If $r$ is even, then $\C_{D(f)}$ is a two-weight code over $\gf(p)$ with parameters $[(q-1)/e,\, m,\, (p-1)(q-p^{m-r/2})/ep]$ 
and weight enumerator 
\begin{eqnarray}
1+ \frac{q-1}{2} z^{(p-1)(q-p^{m-r/2})/ep} + \frac{q-1}{2} z^{(p-1)(q+p^{m-r/2})/ep}. 
\end{eqnarray} 
\end{theorem}

\begin{proof}
Since $f$ is $e$-to-$1$ in $\gf(q)^*$, we have that $n_f=|D(f)|=(q-1)/e$ and 
\begin{eqnarray*}
\sum_{y \in \gf(p)^*} \sum_{x \in \gf(q)} \chi_1(y f(x)) 
&=& \sum_{y \in \gf(p)^*} \left(1 + \sum_{x \in \gf(q)^*} \chi_1(y f(x)) \right) \\
&=& p-1 + \sum_{y \in \gf(p)^*}  \sum_{x \in \gf(q)^*} \chi_1(y f(x))  \\
&=& p-1 + e\sum_{y \in \gf(p)^*}  \chi_1(y D(f)). 
\end{eqnarray*}
It then follows from Lemma \ref{lem-ZD13} that 
\begin{eqnarray*}
\sum_{y \in \gf(p)^*}  \chi_1(y D(f)) = 
\left\{ 
\begin{array}{ll}
-\frac{(p-1)(1\pm p^{m-r/2})}{e} & \mbox{ if $r$ even} \\
-\frac{p-1}{e} & \mbox{ if $r$ odd.} 
\end{array}
\right. 
\end{eqnarray*}

Note that $uf$ is also a quadratic form over $\gf(q)$ and satisfies all the conditions of Theorem 
\ref{thm-qfcodes} for every $u \in \gf(q)^*$.  
It then follows from (\ref{eqn-weight}) that the Hamming weight of the codeword $\bc_u$ is 
given by 
\begin{eqnarray*}
\wt(\bc_u)=\left\{ 
\begin{array}{ll}
\frac{(p-1)(q \pm p^{m-r/2})}{ep} & \mbox{ if $r$ even} \\
\frac{(p-1)q}{ep} & \mbox{ if $r$ odd,} 
\end{array}
\right. 
\end{eqnarray*}
where $u \in \gf(q)^*$. Hence, the dimension of the code is $m$, as $\bc_u>0$ for each $u \in \gf(q)^*$. 

If $r$ is odd, the code is a one-weight code with the nonzero weight $(p-1)q/ep$. If $r$ is even, the code has 
the following nonzero weights 
$$ 
w_i = \frac{(p-1)(q +(-1)^i p^{m-r/2})}{ep} 
$$
for $i \in \{1,2\}$. 
Since $0 \not\in D(f)$, the minimum distance of the dual code $\C_{D(f)}^\perp$ is at least 2. 
The first two Pless Power Moments \cite[p.260]{HP} lead to the following system of equations:  
\begin{eqnarray*}
\left\{ 
\begin{array}{lll}
A_{w_1}+A_{w_2} &=& p^m-1, \\
w_1A_{w_1}+w_2A_{w_2} &=& (p-1)p^{m-1}n_f . 
\end{array}
\right. 
\end{eqnarray*} 
Solving this system of equations gives the desired weight distribution for the case $r$ being even. 
This completes the proof. 
\end{proof}

Let $D$ be the set of all nonzero squares in $\gf(q)$. It is well known that $D$ is a $(q, \,(q-1)/2, \,(q-3)/2)$ difference set 
in $(\gf(q), \,+)$ if $q \equiv 3 \pmod{4}$, and an almost difference set with parameters 
$\left(q, \,(q-1)/2, \,(q-5)/4, \,(q-1)/2\right)$ if $q \equiv 1 \pmod{4}$. 

As an special case of Theorem \ref{thm-qfcodes}, we have the following.  

\begin{corollary}\label{thm-part1}
Let $D$ be the set of all quadratic residues in $\gf(p^m)^*$. 
If $m$ is odd, then $\C_D$ is a one-weight code over $\gf(p)$ with parameters 
$[(q-1)/2, \,m, \,(p-1)q/2p]$. 

If $m$ is even, then $\C_D$ is a two-weight code over $\gf(p)$ with parameters $[(q-1)/2,\, m,\, (p-1)(q-\sqrt{q})/2p]$ 
and weight enumerator 
\begin{eqnarray}
1+ \frac{q-1}{2} z^{(p-1)(q-\sqrt{q})/2p} + \frac{q-1}{2} z^{(p-1)(q+\sqrt{q})/2p}. 
\end{eqnarray} 
\end{corollary}

\begin{proof}
Note that $f(x)=x^2$ is a quadratic form of rank $m$ over $\gf(q)$ satisfying the conditions of Theorem \ref{thm-qfcodes}, 
with $e=2$. The desired conclusions follow from Theorem \ref{thm-qfcodes}. 
\end{proof}

To obtain more classes of one-weight and two-weight codes from Theorem \ref{thm-qfcodes}, we need to 
find quadratic forms $f$ over $\gf(q)$ satisfying the conditions of Theorem \ref{thm-qfcodes}. Below are 
a few more examples. 

\begin{example} 
$f(x)=x^{p^\ell+1}$ is a quadratic form over $\gf(q)$ satisfying the conditions of Theorem \ref{thm-qfcodes}, 
where $e=\gcd(q-1, p^\ell +1)$. 
\end{example}  

\begin{example} 
$f(x)=x^{10} -ux^{6} -u^2x^2$ is a quadratic form over $\gf(3^m)$ satisfying the conditions of Theorem \ref{thm-qfcodes}, 
where $u \in \gf(3^m)$, $m$ is odd, and $e=2$. 
\end{example}

\subsection{The codes $\C_{D(f)}$ from the images of some quadratic functions on $\gf(2^m)$}

An \emph{$h$-arc\index{$h$-arc}} in the projective plane $\PG(2,q)$, with $q$ a prime power, is a set 
of $h$-points such that no three of them are collinear. The maximum value for $h$ is $q+1$ if $q$ is 
odd, and $q+2$ if $q$ is even.  If $q$ is odd, $(q+1)$-arcs are called \emph{ovals.\index{ovals}} If $q$ 
is even, $(q+2)$-arcs are called \emph{hyperovals.\index{hyperovals}}

In 1998, Maschietti discovered a connection between hyperoval sets and difference sets and proved the 
following result \cite{Masch}. 

\begin{theorem} 
Let $m$ be odd and let $n=2^m-1$. Then 
$$ 
D_\rho:=  \{x^\rho+x: \,x \in \gf(2^m)\} \setminus \{0\}   
$$  
is a difference set with Singer parameters $(2^m-1, \,2^{m-1}-1, \,2^{m-2}-1)$ in $(\gf(2^m)^*, \,\times)$ 
if $x \mapsto x^\rho$ is a permutation on $\gf(2^m)$ and the mapping $\Gamma_\rho: x \mapsto x^\rho+x$ is two-to-one 
on $\gf(2^m)$. 

In particular, the following $\rho$ 
yields difference sets: 
\begin{itemize} 
\item $\rho=2$ (Singer case). 
\item $\rho=6$ (Segre case). 
\item $\rho=2^\sigma + 2^\pi$ with $\sigma=(m+1)/2$ and $4\pi \equiv 1 \bmod{m}$ (Glynn I case). 
\item $\rho=3 \cdot 2^\sigma + 4$ with $\sigma=(m+1)/2$ (Glynn II case).    
\end{itemize} 
\end{theorem}

In the rest of this section, let $m$ be odd. Let $\rho=2^i + 2^j$, 
where $i$ and $j$ are nonnegative integers such that $0 \leq i < j<m$.  Let $\kappa=j-i$. Define 
$$ 
\image(\Gamma_\rho)=\{\Gamma_\rho(x): \,x \in \gf(2^m)\}. 
$$ 
We now study the code $\C_{D(\Gamma_\rho)}$ when $\rho=2^i + 2^j$ and the mapping $\Gamma_\rho(x)=x^\rho+x$ 
satisfies certain conditions. To this end, we define the following Boolean function from $\gf(2^m)$ 
to $\gf(2)$: 
\begin{eqnarray}
f(x)=\left\{ \begin{array}{ll}
1 \mbox{ if } x \in \image(\Gamma_\rho), \\
0 \mbox{ otherwise.} 
\end{array}
\right. 
\end{eqnarray}

The following lemma is proved in \cite{Xiang99}. 

\begin{lemma}\label{lem-Xiang}
Let $\rho=2^i + 2^j$. If $\Gamma_\rho$ is two-to-one on $\gf(2^m)$ and $\gcd(2^\kappa+1, \,2^m-1)=1$, where $\kappa=j-i$, 
then
\begin{eqnarray}
\hat{f}(b)=\left\{ \begin{array}{ll}
                           0 & \mbox{ if } b =0, \\
                           0 & \mbox{ if } b \ne 0, \ \tr(b^\ell)=0, \\
                           \pm 2^{(m+1)/2} &  \mbox{ if } b \ne 0, \ \tr(b^\ell)=1, 
\end{array}
\right. 
\end{eqnarray} 
where $\hat{f}$ denotes the Walsh transform of $f$ and 
$$ 
\ell = \frac{2^i+2^j-1}{2^\kappa +1}. 
$$

Furthermore, if $\tr(b^\ell)=1$, 
$$ 
\hat{f}(b)=(-1)^{\tr(u+u^{2^\kappa +1})} \hat{f}(1), 
$$
where $b^\ell=1+u^{2^{j+\kappa}} +u^{2^{j-\kappa}}$. 
\end{lemma} 

\begin{table}[ht]
\begin{center} 
\caption{The weight distribution of the codes of Theorem \ref{thm-hyperovalDS}}\label{tab-semibentfcode6}
\begin{tabular}{|c|c|} \hline
Weight $w$ &  Multiplicity $A_w$  \\ \hline  
$0$          &  $1$ \\ \hline 
$2^{m-2}-2^{(m-3)/2}$ & $2^{m-2}+2^{(m-3)/2}$ \\ \hline 
$2^{m-2}$ & $2^{m-1}-1$ \\ \hline 
$2^{m-2}+2^{(m-3)/2}$ & $2^{m-2}-2^{(m-3)/2}$ \\ \hline 
\end{tabular}
\end{center} 
\end{table} 

\begin{theorem}\label{thm-hyperovalDS}
Let $\rho=2^i + 2^j$. If $\Gamma_\rho$ is two-to-one on $\gf(2^m)$ and $\gcd(2^\kappa+1, \,2^m-1)=1$, where $\kappa=j-i$, 
then the binary code $\C_{D(\Gamma_\rho)}$ has parameters $[2^{m-1}-1, \,m, \,2^{m-2} - 2^{(m-3)/2}]$ and the weight 
distribution of Table \ref{tab-semibentfcode6}.
\end{theorem}

\begin{proof} 
Let $n=2^{m-1}-1$. 
Let $b \in \gf(2^m)^*$ and $\bc_b$ be the codeword of (\ref{eqn-mcodeword}) in $\C_{D(\Gamma_\rho)}$. We now determine 
the Hamming weight of $\bc_b$. 

By definition, we have 
\begin{eqnarray*}
\hat{f}(b) &=& \sum_{x \in \gf(2^m)} (-1)^{f(x) + \tr(bx)} \\
&=& - \sum_{x \in \image(\Gamma_\rho)} (-1)^{\tr(bx)} + 
           \sum_{x \in \gf(2^m) \setminus \image(\Gamma_\rho)} (-1)^{\tr(bx)} \\
&=& - 2 \sum_{x \in \image(\Gamma_\rho)} (-1)^{\tr(bx)} \\
&=& - 2 \left(\sum_{x \in D(\Gamma_\rho)} (-1)^{\tr(bx)}  +1\right) \\
\end{eqnarray*}

It then follows from Lemma \ref{lem-Xiang} that 
$$ 
\chi_b(D(\Gamma_\rho)) \in \{-1, \,\pm 2^{(m-1)/2} -1\}. 
$$
The desired conclusions on the weights and the dimension of this code  $\C_{D(\Gamma_\rho)}$ follow from (\ref{eqn-weight}). 

It is easily seen that the minimum weight of the dual code $\C_{D(\Gamma_\rho)}^\perp$ is at least $3$. Define 
$$ 
w_1=2^{m-2}-2^{(m-3)/2}, \ w_2=2^{m-2}, \ w_3=2^{m-2}+2^{(m-3)/2}.  
$$
We now determine the number $A_{w_i}$ of codewords with weight $w_i$ in $\C_{D(\Gamma_\rho)}$. 
The first three Pless Power Moments \cite[p.260]{HP} lead to the following system of equations:  
\begin{eqnarray}\label{eqn-wtdsemibentfcode6}
\left\{ 
\begin{array}{lll}
A_{w_1}+A_{w_2}+A_{w_3} &=& 2^m-1, \\
w_1A_{w_1}+w_2A_{w_2}+w_3A_{w_3} &=& n 2^{m-1}, \\ 
w_1^2A_{w_1}+w_2^2A_{w_2}+w_3^2A_{w_3} &=& n(n+1) 2^{m-2}.  
\end{array}
\right. 
\end{eqnarray} 
Solving this system of equations gives the desired weight distribution. 
This completes the proof. 

\end{proof} 

In the Segre case, we have $\rho=2^1 + 2^2$, $(i, \,j)=(1,2)$, and $\kappa =1$. Hence $\gcd(2^\kappa +1, \,2^m-1)=1$. 
It is known that $\Gamma_\rho$ is two-to-one. Hence, all the conditions in Theorem \ref{thm-hyperovalDS} are satisfied. 
Thus, the difference set $D_6$ gives a class of bianry linear codes with three weights.   

The following lemma can be easily proved. 

\begin{lemma}\label{lem-dcs} 
Define 
\begin{eqnarray}
\kappa = \left\{ \begin{array}{ll}
    \frac{m+1}{4} & \mbox{ if } m \equiv 3 \pmod{4}, \\
    \frac{m-1}{4} & \mbox{ if } m \equiv 1 \pmod{4}.      
\end{array}
\right. 
\end{eqnarray} 
Then $\gcd(2^\kappa +1, \,2^m-1)=1$. 
\end{lemma} 

In the Glynn I case, $(i, \,j, \,\kappa)=((m+1)/4, \,(m+1)/2, \,(m+1)/4)$ if  $m \equiv 3 \pmod{4}$, 
and $(i, \,j, \,\kappa)=((m+1)/2, \,(3m+1)/4, \,(m-1)/4)$ if  $m \equiv 1 \pmod{4}$. In this case, 
we have $\gcd(2^\kappa +1, \,2^m-1)=1$ by Lemma \ref{lem-dcs}. It is known that, in the 
Glynn I case, $\Gamma_\rho$ is two-to-one. Hence, all the conditions in Theorem \ref{thm-hyperovalDS} are satisfied. 
Thus, the difference set $D_{2^i +2^j}$ gives a class of binary linear codes with three weights.    

In the Glynn II case, $\rho=3 \cdot 2^\sigma + 4$ and the mapping $\Gamma_\rho$ is not quadratic. 
It looks hard to determine the weight distribution of the code $\C_{D(\Gamma_\rho)}$. 
When $m=5$, the binary code $\C_{D(\Gamma_\rho)}$ in the  Glynn II case has parameters $[15, \,5, \,6]$ and 
the weight enumerator $1+10z^6+15z^8+6z^{10}$. When $m=7$, the binary code $\C_{D(\Gamma_\rho)}$ in the  Glynn II case has parameters $[63, \,7, \,28]$ and 
the weight enumerator $1+36z^{28}+63z^{32}+28z^{36}$. When $m \geq 9$, we have the following conjecture. 

\begin{conj} 
Let $\rho=3 \cdot 2^{(m+1)/2} + 4$, where $m \geq 9$ and $m$ is odd. Then the binary linear code $\C_{D(\Gamma_\rho)}$ in the  Glynn II case has parameters 
$[2^{m-1}-1, \,m, \,2^{m-2}-2^{(m-1)/2}]$ and has only the following five nonzero weights: 
$$ 
2^{m-2}, \,2^{m-2} \pm 2^{(m-1)/2}, \,2^{m-2} \pm 2^{(m-3)/2}. 
$$ 
\end{conj} 

\begin{example} 
When $m=9$, the binary code $\C_{D(\Gamma_\rho)}$ in the  Glynn II case has parameters $[255, \, 9, \, 112]$ and weight enumerator 
$  
1+9z^{112}+108z^{120}+285z^{128}+108z^{136}+z^{144}. 
$ 
\end{example} 

\begin{example} 
When $m=11$, the binary code $\C_{D(\Gamma_\rho)}$ in the  Glynn II case has parameters $[1023, \, 11, \, 480]$ and weight enumerator 
$  
1+22z^{480}+440z^{496}+1155z^{512}+408z^{528}+22z^{544}. 
$ 
\end{example} 

It would be nice if the weight distribution 
of the binary linear code $\C_{D(\Gamma_\rho)}$ in the  Glynn II case can be determined.

\section{Linear codes from the preimage $f^{-1}(b)$ for functions $f$ from $\gf(p^m)$ to $\gf(p)$}

Let $f$ be a function from $\gf(p^m)$ to $\gf(p)$, and let $D$ be any subset of the preimage $f^{-1}(b)$ for any 
$b \in \gf(p)$. In this section, we consider the code $\C_{D}$.  Similarly, it is very hard to determine the parameters 
of this code in general. We shall deal with a few special cases in this section.   

\subsection{The Boolean case} 

Let $f$ be a Boolean function from $\gf(2^m)$ to $\gf(2)$. The \emph{support} of $f$ is defined to be 
$$
D_f=\{x \in\gf(2^m) : f(x)=1\} \subseteq \gf(2^m). 
$$ 
Recall that $n_f=|D_f|$.  

The {\em Walsh transform} of $f$ is defined by 
\begin{eqnarray}\label{eqn-WalshTransform2}
\hat{f}(w)=\sum_{x \in \gf(2^m)} (-1)^{f(x)+\tr(wx)} 
\end{eqnarray} 
where $w \in \gf(2^m)$. The {\em Walsh spectrum} of $f$ is the following multiset 
$$ 
\left\{\left\{ \hat{f}(w): w \in \gf(2^m) \right\}\right\}. 
$$ 

In this section, we investigate the binary code $\C_{D_f}$ with length $n_f$ 
and dimension at most $m$, and will determine the weight distribution of the code $\C_{D_f}$ for several classes 
of Boolean functions $f$ whose supports $D_f$ are certain 2-designs.  

A function $f$ from $\gf(2^m)$ to $\gf(2)$ is called {\em linear} if $f(x+y)=f(x)+f(y)$ for all $(x, y) \in \gf(2^m)^2$. 
A function $f$ from $\gf(2^m)$ to $\gf(2)$ is called {\em affine} if $f$ or $f-1$ is affine. 

The main result of this section is described in the following theorem. 

\begin{theorem}\label{thm-BooleanCodes}
Let symbols and notation be as above. If $f$ is not an affine function, then $\C_{D_f}$ is a binary linear code with 
length $n_f$ and dimension $m$, and its weight distribution is given by the following multiset: 
\begin{eqnarray}\label{eqn-WTBcodes}
\left\{\left\{ \frac{2n_f+\hat{f}(w)}{4}: w \in \gf(2^m)^*\right\}\right\} \cup \left\{\left\{ 0 \right\}\right\}. 
\end{eqnarray} 
\end{theorem}

\begin{proof}
Note that all characters of the group $(\gf(2^m),\,+)$ are of the form 
$$ 
\chi_b(x)=(-1)^{\tr(bx)},  \ b \in \gf(2^m). 
$$

Let $w \in \gf(2^m)^*$. 
It follows from (\ref{eqn-WalshTransform2}) that 
\begin{eqnarray*}
\hat{f}(w) = - \sum_{ x \in D_f} \chi_w(x) + \sum_{ x \in \gf(2^m) \setminus D_f} \chi_w(x) 
               = -2 \sum_{ x \in D_f} \chi_w(x) = -2 \chi_w(D_f). 
\end{eqnarray*} 
It then follows from (\ref{eqn-weight}) that the Hamming weight of the codeword $\bc_w$ of (\ref{eqn-mcodeword}) 
is equal to $(2n_f+\hat{f}(w))/4$. Hence, the weight distribution of $\C_{D_f}$ is given by the multiset in (\ref{eqn-WTBcodes}). 
Since $f$ is not affine, $\bc_w >0$ for every nonzero $w \in \gf(2^m)$. Thus, the dimension of $\C_{D_f}$ is equal to $n_f$. 
This completes the proof.  
\end{proof}

Theorem \ref{thm-BooleanCodes} establishes a connection between Boolean functions and a class of linear codes. 
The determination of the weight distribution of the binary linear code $\C_{D_f}$ is equivalent to that of the Walsh 
spectrum of the Boolean function $f$. When the Boolean function $f$ is selected properly, the code $\C_{D_f}$ has 
only a few weights and may have good parameters. We will demonstrate this in the remainder of this section.

\subsubsection{Linear codes from bent functions} 

A function from $\gf(2^m)$ to $\gf(2)$ is called \emph{bent\index{bent}} if $|\hat{f}(w)|=
2^{m/2}$ for every $w \in \gf(2^m)$. Bent functions exist only for even $m$, and were coined 
by Rothaus in \cite{Rothaus76}.  

It is well known that 
a function $f$ from $\gf(2^m)$ to $\gf(2)$ is bent if and only if $D_f$ is 
a difference set in  $(\gf(2^m),\,+)$ with the following parameters 
\begin{eqnarray}\label{eqn-MenonHadamardPara}
(2^m, \,2^{m-1} \pm 2^{(m-2)/2}, \,2^{m-2} \pm 2^{(m-2)/2}).  
\end{eqnarray} 

Let $f$ be bent. Then by definition $\hat{f}(0)=\pm 2^{m/2}$. It then follows that 
\begin{eqnarray}\label{eqn-bentfuncsupportsize}
n_f=|D_f|=2^{m-1} \pm 2^{(m-2)/2}
\end{eqnarray}

\begin{table}[ht]
\begin{center} 
\caption{The weight distribution of the codes of Corollary \ref{thm-bentcodes}}\label{tab-bentfcode}
\begin{tabular}{|c|c|} \hline
Weight $w$ &  Multiplicity $A_w$  \\ \hline  
$0$          &  $1$ \\ \hline 
$\frac{n_f}{2}-2^{\frac{m-4}{2}}$          &  $\frac{2^m-1-n_f2^{-\frac{m-2}{2}}}{2}$ \\ \hline   
$\frac{n_f}{2}+2^{\frac{m-4}{2}}$          &  $\frac{2^m-1+n_f2^{-\frac{m-2}{2}}}{2}$ \\ \hline 
\end{tabular}
\end{center} 
\end{table}

As a corollary of Theorem \ref{thm-BooleanCodes}, we have the following. 

\begin{corollary}\label{thm-bentcodes}
Let $f$ be a bent function from $\gf(2^m)$ to $\gf(2)$ with $f(0)=0$, where $m \geq 4$ and is even. Then $\C_{D_f}$ is an 
$[n_f, \,m, \,(n_f-2^{(m-2)/2})/2]$ 
two-weight binary code with the weight distribution in Table \ref{tab-bentfcode}, where $n_f$ is defined in (\ref{eqn-bentfuncsupportsize}). 
\end{corollary} 

\begin{proof}
By Theorem \ref{thm-BooleanCodes}, 
the dimension of the code $\C_{D_f}$ is $m$ as bent functions are not affine. 
It follows from the definition of bent functions and Theorem \ref{thm-BooleanCodes} that $\C_{D_f}$ has nonzero weights 
$(n_f-2^{(m-2)/2})/2$ or $(n_f+2^{(m-2)/2})/2$. 

Let $A_i$ denote the number of codewords with Hamming weight $i$ in $\C_{D_f}$. It is obvious that the dual code 
$\C_{D_f}^\perp$ has minimum weight at least 2. The first two Pless Power Moments \cite[p.260]{HP} lead to the 
following system of equations:  
\begin{eqnarray}\label{eqn-wtdbentfcode}
\left\{ 
\begin{array}{lll}
1 + A_{\frac{n_f}{2}+2^{\frac{m-4}{2}}} +  A_{\frac{n_f}{2}-2^{\frac{m-4}{2}}} &=& 2^m, \\
  \left( \frac{n_f}{2}+2^{\frac{m-4}{2}} \right) A_{\frac{n_f}{2}+2^{\frac{m-4}{2}}} +  
       \left( \frac{n_f}{2}-2^{\frac{m-4}{2}} \right) A_{\frac{n_f}{2}-2^{\frac{m-4}{2}}} &=& n_f2^{m-1}.   
\end{array}
\right. 
\end{eqnarray} 
The desired weight distribution in Table \ref{tab-bentfcode} is obtained by solving (\ref{eqn-wtdbentfcode}). 
\end{proof}

\begin{example} 
Let $m=6$ and let $f$ be a bent function from $\gf(2^6)$ to $\gf(2)$ with $|D_f|=2^{6-1} - 2^{(6-2)/2}=28$. 
Then the code $\C_{D_f}$ has parameters $[28, \,6, \,12]$ and is optimal. 
\end{example}

\begin{example} 
Let $m=8$ and let $f$ be a bent function from $\gf(2^8)$ to $\gf(2)$ with $|D_f|=2^{8-1} - 2^{(8-2)/2}=120$. 
Then the code $\C_{D_f}$ has parameters $[120, \,8, \,56]$, while the optimal binary code has parameters $[120, \,8, \,58]$.  
\end{example} 

There are many constructions of bent functions and thus Hadamard difference sets. We refer the reader to 
\cite{BCHKM}, \cite{Mesnager142}, \cite{Mesnager143},  the book chapter \cite{Carlet} and the references 
therein for details. Any bent function can be plugged into Corollary \ref{thm-bentcodes} to obtain a two-weight linear code. 

\subsubsection{Linear codes from semibent functions} 

Let $m$ be odd. Then there is no bent Boolean function on $\gf(2^m)$.  
A function $f$ from $\gf(2^m)$ to $\gf(2)$ is called \emph{semibent} if $\hat{f}(w) \in \{0, \,
\pm 2^{(m+1)/2}\}$ for every $w \in \gf(2^m)$. 

Let $f$ be a semibent function from $\gf(2^m)$ to $\gf(2)$. It then follows from the definition of semibent functions 
that  
\begin{eqnarray}\label{eqn-semibf}
n_f=|D_f| &=& \left\{ \begin{array}{ll}
                           2^{m-1}-2^{(m-1)/2} & \mbox{ if } \hat{f}(0)=2^{(m+1)/2}, \\
                           2^{m-1}+2^{(m-1)/2} & \mbox{ if } \hat{f}(0)=-2^{(m+1)/2}, \\       
                           2^{m-1}  & \mbox{ if } \hat{f}(0)=0.                                                
\end{array}
\right. 
\end{eqnarray}

\begin{table}[ht]
\begin{center} 
\caption{The weight distribution of the codes of Corollary \ref{thm-semibentcodes}}\label{tab-semibentfcode}
\begin{tabular}{|c|c|} \hline
Weight $w$ &  Multiplicity $A_w$  \\ \hline  
$0$          &  $1$ \\ \hline 
$\frac{n_f-2^{(m-1)/2}}{2}$ & $n_f(2^m-n_f)2^{-m} - n_f2^{-(m+1)/2}$ \\ \hline 
$\frac{n_f}{2}$ & $2^m-1-n_f(2^m-n_f)2^{-(m-1)}$ \\ \hline 
$\frac{n_f+2^{(m-1)/2}}{2}$ & $n_f(2^m-n_f)2^{-m} + n_f2^{-(m+1)/2}$ \\ \hline 
\end{tabular}
\end{center} 
\end{table} 

A semibent function $f$ gives a 2-design. The reader is referred to \cite{DN} for details of the $2$-design. 
We are interested in the coding theory aspect of semibent functions. 

As a corollary of Theorem \ref{thm-BooleanCodes}, we have the following. 

\begin{corollary}\label{thm-semibentcodes}
Let $f$ be a semibent function from $\gf(2^m)$ to $\gf(2)$ with $f(0)=0$, where $m$ is odd. 
Then $\C_{D_f}$ is an $[n_f, \,m, \,(n_f-2^{(m-1)/2})/2]$ 
three-weight binary code with the weight distribution in Table \ref{tab-semibentfcode}, where $n_f$ is defined in (\ref{eqn-semibf}). 
\end{corollary} 

\begin{proof}
By Theorem \ref{thm-BooleanCodes}, 
the dimension of the code $\C_{D_f}$ is $m$ as semibent functions are not affine. 
It follows from the definition of semibent functions and Theorem \ref{thm-BooleanCodes} that $\C_{D_f}$ has nonzero weights: 
$$ 
w_1=\frac{n_f-2^{(m-1)/2}}{2}, \ w_2=\frac{n_f}{2}, \ w_3=\frac{n_f+2^{(m-1)/2}}{2}.  
$$
We now determine the number $A_{w_i}$ of codewords with weight $w_i$ in $\C_{D_f}$. 
It is straightforward to see that the minimum weight of the dual code $\C_{D_f}^\perp$ is at least $3$.  
The first three Pless Power Moments \cite[p.260]{HP} lead to the following system of equations:  
\begin{eqnarray}\label{eqn-wtdsemibentfcode}
\left\{ 
\begin{array}{lll}
A_{w_1}+A_{w_2}+A_{w_3} &=& 2^m-1, \\
w_1A_{w_1}+w_2A_{w_2}+w_3A_{w_3} &=& n_f 2^{m-1}, \\ 
w_1^2A_{w_1}+w_2^2A_{w_2}+w_3^2A_{w_3} &=& n_f(n_f+1) 2^{m-2}.  
\end{array}
\right. 
\end{eqnarray} 
Solving this system of equations gives the desired weight distribution. 
This completes the proof. 
\end{proof}

\begin{example} 
Let $m=7$ and let $f$ be a semibent function from $\gf(2^7)$ to $\gf(2)$ with $|D_f|=2^{7-1} - 2^{(7-1)/2}=56$. 
Then the code $\C_{D_f}$ has parameters $[56, \,7, \,24]$, while the optimal binary code has parameters $[56, \,7, \,26]$. 
\end{example}

There are a lot of constructions of semibent functions from $\gf(2^m)$ to $\gf(2)$. We refer the reader to 
\cite{CM12,CM14, DQFL,Mesnager11,Mesnager13,Mesnager14} for detailed constructions. All semibent functions can be plugged into Corollary \ref{thm-semibentcodes} to obtain three-weight binary linear codes. 

\subsubsection{Linear codes from almost bent functions}

For any function $g$ from $\gf(2^m)$ to $\gf(2^m)$, we define 
$$ 
\lambda_g(a, b) = \sum_{x \in \gf(2^m)} (-1)^{\tr(ag(x)+bx)}, \ a,\, b \in \gf(2^m).     
$$ 
A function $g$ from $\gf(2^m)$ to $\gf(2^m)$ is called {\em almost bent} if 
$\lambda_g(a, b) = 0, \mbox{ or } \pm 2^{(m+1)/2}$ for every pair $(a, b)$ with $a \neq 0$. 
By definition, almost bent functions over $\gf(2^m)$ exist only for odd $m$.  
Specific almost bent functions are available in \cite{BCP,Carlet}.

By definition,  $\lambda_g(1, 0) \in \{0, \,\pm 2^{(m+1)/2}\}$ for any almost bent function $g$ on $\gf(2^m)$. It is straightforward 
to deduce the following lemma. 

\begin{lemma}\label{lem-absize}
For any almost bent function $g$ from $\gf(2^m)$ to $\gf(2^m)$, define $f=\tr(g)$. Then we have  
\begin{eqnarray}\label{eqn-newnf2}
n_{f}=|D_{\tr(g)}| &=& \left\{ \begin{array}{ll}
                           2^{m-1}+2^{(m-1)/2} & \mbox{ if } \lambda_g(1, 0)=-2^{(m+1)/2}, \\
                           2^{m-1}-2^{(m-1)/2} & \mbox{ if } \lambda_g(1, 0)=2^{(m+1)/2}, \\       
                           2^{m-1}  & \mbox{ if } \lambda_g(1, 0)=0.                                                
\end{array}
\right. 
\end{eqnarray} 
\end{lemma}

As a corollary of Theorem \ref{thm-BooleanCodes}, we have the following. 

\begin{corollary}\label{thm-abcodes}
Let $g$ be an almost bent function from $\gf(2^m)$ to $\gf(2^m)$ with $\tr(g(0))=0$, where $m$ is odd. Define $f=\tr(g)$. 
Then $\C_{D_{f}}$ is an $[n_f, \,m, \,(n_f-2^{(m-1)/2})/2]$ 
three-weight binary code with the weight distribution in Table \ref{tab-semibentfcode}, where $n_f$ is given in (\ref{eqn-newnf2}). 
\end{corollary} 

\begin{proof}
By Theorem \ref{thm-BooleanCodes}, 
the dimension of the code $\C_{D_f}$ is $m$ as $f=\tr(g)$ is not affine. 
It follows from the definition of almost bent functions and Theorem \ref{thm-BooleanCodes} that $\C_{D_f}$ has nonzero weights: 
$$ 
\frac{n_f-2^{(m-1)/2}}{2}, \ \frac{n_f}{2}, \ \frac{n_f+2^{(m-1)/2}}{2}. 
$$
It is easy to prove that the dual code $\C_{D_f}$ has minimum weight at least 3. 
The frequencies of the three weights are already determined in the proof of Corollary \ref{thm-semibentcodes}. 
This completes the proof.  
\end{proof}

We remark that the binary code  $\C_{D_f}$ of Corollarey \ref{thm-abcodes} is different from the code 
from almost bent functions defined in \cite{CCZ}, as the dimensions and lengths of the codes are different.  

\subsubsection{Linear codes from quadratic Boolean functions} 

Let 
\begin{eqnarray}\label{eqn-QBFs}
f(x)=\tr_{2^m/2} \left(  \sum_{i=0}^{\lfloor m/2 \rfloor} f_i x^{2^i +1} \right) 
\end{eqnarray}
be a quadratic Boolean function from $\gf(2^m)$ to $\gf(2)$, where $f_i \in \gf(2^m)$. Similarly, the rank of 
$f$, denoted by $r_f$, is defined to be the codimension of the $\gf(2)$-vector space 
$$ 
V_f=\{x \in \gf(2^m): f(x+z)-f(x)-f(z)=0 \ \forall \ z \in \gf(2^m)\}. 
$$
The Walsh spectrum of $f$ is known \cite{CCK} and is given in Table \ref{tab-WalshBBFs}.  

\begin{table}[ht]
\begin{center} 
\caption{The Walsh spectrum of quadratic Boolean functions}\label{tab-WalshBBFs}
\begin{tabular}{|c|c|} \hline
$\hat{f}(w)$ &  the number of $w$'s  \\ \hline  
$0$          &  $2^m-2^{r_f}$ \\ \hline 
$2^{m-r_f/2}$          &  $2^{r_f-1}+2^{(r_f-2)/2}$ \\ \hline 
$-2^{m-r_f/2}$          &  $2^{r_f-1}-2^{(r_f-2)/2}$ \\ \hline 
\end{tabular}
\end{center} 
\end{table} 

Let $D_f$ be the support of $f$. By definition, we have 
\begin{eqnarray}\label{eqn-QBFcodeL}
n_f=|D_f|=2^{m-1}-\frac{\hat{f}(0)}{2} 
=\left\{ 
\begin{array}{ll}
2^{m-1}   & \mbox{ if } \hat{f}(0)=0, \\
2^{m-1}-2^{m-1-r_f/2}   & \mbox{ if } \hat{f}(0)=2^{m-1-r_f/2}, \\
2^{m-1}+2^{m-1-r_f/2}   & \mbox{ if } \hat{f}(0)=-2^{m-1-r_f/2}. 
\end{array}
\right. 
\end{eqnarray}

The following theorem then follows from Theorem \ref{thm-BooleanCodes} and Table \ref{tab-WalshBBFs}. 

\begin{theorem}\label{thm-CodeQBFs}
Let $f$ be a quadratic Boolean function of the form in (\ref{eqn-QBFs}) and $f \neq 0$. Then 
$\C_{D_f}$ is a binary code with length $n_f$ given in (\ref{eqn-QBFcodeL}), dimension $m$, 
and the weight distribution in Table \ref{tab-WEqbfs}, where 
\begin{eqnarray}
(\epsilon_1, \epsilon_1, \epsilon_3)=
\left\{ 
\begin{array}{ll}
(1,0,0)   & \mbox{ if } \hat{f}(0)=0, \\
(0,1,0)   & \mbox{ if } \hat{f}(0)=2^{m-1-r_f/2}, \\
(0,0,1)   & \mbox{ if } \hat{f}(0)=-2^{m-1-r_f/2}. 
\end{array}
\right. 
\end{eqnarray} 
\end{theorem}

\begin{table}[ht]
\begin{center} 
\caption{The weight distribution of the code $\C_{D_{f}}$ in Theorem \ref{thm-CodeQBFs}}\label{tab-WEqbfs}
\begin{tabular}{|c|c|} \hline
Weight $w$ &  $A_w$  \\ \hline  
$0$          &  $1$ \\ \hline
$\frac{n_f}{2}$          &  $2^m-2^{r_f}-\epsilon_1$ \\ \hline 
$\frac{n_f+2^{m-1-r_f/2}}{2}$          &  $2^{r_f-1}+2^{(r_f-2)/2}-\epsilon_2$ \\ \hline 
$\frac{n_f-2^{m-1-r_f/2}}{2}$          &  $2^{r_f-1}-2^{(r_f-2)/2}-\epsilon_3$ \\ \hline 
\end{tabular}
\end{center} 
\end{table} 

Note that the code $\C_{D_f}$ in Theorem \ref{thm-CodeQBFs} defined by any quadratic Boolean function $f$ is 
different from any subcode of the second-order Reed-muller code, due to the difference in their lengths. The weight 
distributions of the two codes are also different.   

\subsection{A ternary case} 

In this subsection, we analyse a class of ternary codes whose defining sets are  a family of cyclic difference sets, 
which are described in the following theorem \cite{HKM01}.  

\begin{theorem}\label{thm-HDMds}
Let $m=3h \geq 3$ for some positive integer $h$ and $\ell=3^{2h}-3^h+1$. Define $n=(3^m-1)/2$ and 
\begin{eqnarray}\label{eqn-HKMds}
D=\left\{ \alpha^t: \tr_{3^m/3}(\alpha^t + \alpha^{t\ell})=0, \ 0 \leq t \leq n-1  \right\}, 
\end{eqnarray}
where $\alpha$ is a generator of $\gf(3^m)^*$. Then $D$ is a difference set  in $(\gf(3^m)^*/\gf(2)^*, \,\times)$ with the 
following parameters 
\begin{eqnarray}\label{eqn-Singerp=3Param}
\left( \frac{3^m-1}{3-1}, \,\frac{3^{m-1}-1}{3-1}, \,\frac{3^{m-2}-1}{3-1}  \right). 
\end{eqnarray}  
\end{theorem} 

Let $f(x)=\tr_{3^m/3}(x+x^\ell)$, a function from $\gf(3^m)$ to $\gf(3)$. When $h$ is odd, $\{D, -D, \{0\}\}$ forms a partition of the preimage $f^{-1}(0)$.

\begin{table}[ht]
\begin{center} 
\caption{The weight distribution of the codes of Theorem \ref{thm-HKMcodes}}\label{tab-HKMcodes}
\begin{tabular}{|c|c|} \hline
Weight $w$ &  Multiplicity $A_w$  \\ \hline  
$0$          &  $1$ \\ \hline 
$3^{3h-2}-3^{2h-2}$  & $3^{2h}+3^{h}$ \\ \hline 
$3^{3h-2}$  & $3^{3h}-2 \times 3^{2h}-1$ \\ \hline 
$3^{3h-2}+3^{2h-2}$  & $3^{2h}-3^{h}$ \\ \hline 
\end{tabular}
\end{center} 
\end{table} 

Our main result of this section is the following. 

\begin{theorem}\label{thm-HKMcodes} 
Let $h$ be an odd positive integer and let $m=3h$. Let $D$ be defined as in (\ref{eqn-HKMds}). Then the ternary code 
$\C_D$ has parameters 
$$ 
\left[ \frac{3^{3h-1}-1}{2}, \,3h, \,3^{3h-2}-3^{2h-2}\right] 
$$
and the weight distribution of Table \ref{tab-HKMcodes}. 
\end{theorem} 

We remark that the code $\C_D$ of Theorem \ref{thm-HKMcodes} has more than three nonzero weights if $h$ is 
even. To prove this theorem, we need to introduce the basics of quadratic forms and prove several lemmas first. 

A {\em quadratic form} $f(X)$ in $m$ variables over $\gf(p)$ is a homogeneous polynomial in $\gf(p)[x_1, \ldots, x_m]$ 
of degree $2$ and can be expressed as 
$$ 
f(X)=\sum_{1 \leq i \leq j \leq m} a_{i,j} x_i x_j, \ \ a_{i,j} \in \gf(p), 
$$
where $X=(x_1, x_2, \ldots, x_m)$. A quadratic form $f(X)$ in $m$ variables over $\gf(p)$ may also be expresses in 
the trace form 
$$ 
f(x)=\tr\left(\sum_{0 \leq i \leq j \leq m-1} x^{p^i + p^j}\right), \ \ x \in \gf(p^m), 
$$
where $\tr$ denotes the trace function from $\gf(p^m)$ to $\gf(p)$. 

The {\em rank} $r_f$ of a quadratic form $f(x)$ over $\gf(p^m)$ is defined as the codimension of the $\gf(p)$-vector space 
$$ 
V_f=\{x \in \gf(p^m): f(x+y)-f(x)-f(y)=0 \mbox{ for all } y \in \gf(p^m)\},
$$ 
i.e., $r_f=m-\dim(V_f)$. 

We shall use the following lemma in the sequel \cite{ZD13}.  

\begin{lemma}\label{lem-ZhouDing} 
Let $f(x)$ be a quadratic form of rank $r_f$ over $\gf(p^m)$, and let $\epsilon_p=e^{2\pi \sqrt{-1}/p}$. Then 
\begin{eqnarray*}
\sum_{y \in \gf(p)^*} \sum_{x \in \gf(p^m)} \epsilon_p^{yf(x)}  = 
\left\{ \begin{array}{ll}
\pm (p-1) p^{m-r_f/2}  & \mbox{ if $r_f$ is even,} \\
0                                  & \mbox{ otherwise.} 
\end{array} 
\right.
\end{eqnarray*}
\end{lemma} 

Starting from now on, we put $p=3$, $m=3h$ and $e=3^h$, where $h$ is odd. We consider the following quadratic form 
$$ 
Q_u(x)=\tr(u x^{e+1} + x^2) 
$$
over $\gf(3^m)$. 
It is easily seen that 
\begin{eqnarray}\label{eqn-bilinear1}
Q_u(y+z)-Q_u(y)-Q_u(z)=\tr((u^{e^2} y^{e^2} + u y^e -y)z). 
\end{eqnarray}
By definition, the rank $r_{Q_u}=m-\log_3(n_u)$, where $n_u$ is the number of solutions $y \in \gf(3^m)$ of the following 
equation 
\begin{equation}\label{eqn-lineareqn}
u^{e^2} y^{e^2} + u y^e -y=0. 
\end{equation}  

The following lemma is proved in \cite{LF08}. A direct proof discussing the number of solutions $n_u$ of 
(\ref{eqn-lineareqn}) can also be given in a straightforward way. 

\begin{lemma}\label{lem-LF08} 
The rank $r_{Q_u}$ of the quadratic form $Q_u$ is $m$, or $m-h$, or $m-2h$. 
\end{lemma} 

We will need the following lemma later. 

\begin{lemma}\label{lem-rankQ1}
The quadratic form $Q_1(y)$ over $\gf(3^m)$ has rank $m=3h$. 
\end{lemma} 

\begin{proof}
It suffices to prove that the equation 
$$ 
y^{e^2} + y^e -y=0 
$$ 
has the only solution $y=0$ in $\gf(3^m)$. Raising both sides of this equation to the power of $e$, we obtain 
$$ 
y+y^{e^2}-y^e=0. 
$$
Adding the two equations above yields $y^{e^2}=0$. Hence, $y=0$. This completes the proof. 
\end{proof}

We shall employ the following lemma whose proof is straightforward and may be found in \cite{LF08}. 

\begin{lemma}\label{lem-LF081}
The rank of the quadratic form $\tr(bx^{e+1})$ is $m$ for all $b \in \gf(3^m)^*$. 
\end{lemma} 

The next lemma will play an important role in determining the weight distribution of the code of Theorem 
\ref{thm-HKMcodes}. 

\begin{lemma}\label{lem-DingCS}
For any $u \in \gf(3^m)$, at least one of the two quadratic forms $Q_u(y)$ and $Q_{-1-u}(y)$ over $\gf(3^m)$ 
has rank $m$, where $m=3h$ and $h$ is odd. 
\end{lemma} 

\begin{proof}
Suppose on the contrary that both $Q_u(y)$ and $Q_{-1-u}(y)$ have rank less than $m$ for some 
$u \in \gf(3^m)$. Then there would exist $y_1 \in \gf(3^m)^*$ and $y_2 \in \gf(3^m)^*$ such that 
\begin{eqnarray}\label{eqn-X1}
\left\{ \begin{array}{r}
u^{e^2} y_1^{e^2} + u y_1^e - y_1 = 0, \\
(1+u)^{e^2} y_2^{e^2} + (1+u) y_2^e + y_2 = 0,   
\end{array}
\right. 
\end{eqnarray}
where $e=3^h$. 

Note that $y^{e^3}=y$ for all $y \in \gf(3^m)$. Raising the first equation of (\ref{eqn-X1}) to the power of 
$e^0$, $e^1$ and $e^2$ yields the following system of three equations 
\begin{eqnarray}\label{eqn-X2}
\left\{ \begin{array}{rrr}
u^{e^2} y_1^{e^2} + u y_1^e - y_1 &=& 0, \\
u y_1 + u^e y_1^{e^2} - y_1^e  &=& 0, \\ 
u^e y_1^e + u^{e^2} y_1 - y_1^{e^2} &=& 0.   
\end{array}
\right. 
\end{eqnarray}
Solving (\ref{eqn-X2}) gives 
\begin{eqnarray}\label{eqn-X3}
u=y_1^{2e^2-e-1}-y_1^{e-1}-y_1^{1-e}. 
\end{eqnarray}

Using the second equation of (\ref{eqn-X1}) in a similar way, we obtain 
\begin{eqnarray}\label{eqn-X4}
u=-y_2^{2e^2-e-1}+y_2^{e-1}+y_2^{1-e}-1. 
\end{eqnarray} 

We now define 
$$ 
P(x)=x^{2e^2-e-1}-x^{e-1}-x^{1-e}-1 \in \gf(3^m)[x]. 
$$
Combining (\ref{eqn-X3}) and (\ref{eqn-X4}), we arrive at 
\begin{eqnarray}\label{eqn-X5}
P(y_1)=-P(y_2). 
\end{eqnarray}

It is straightforward to verify that 
$$ 
P(x)=\frac{(x^{e^2}+x^e-x)^{e+1}}{x^{e+1}}. 
$$
In addition, we have $x^{e^2}+x^e-x \neq 0$ for all $x \in \gf(3^m)^*$. It follows that $P(x)$ is a nonzero 
square in $\gf(3^m)$ for every $x \in \gf(3^m)^*$. Since $m$ is odd, $-1$ is a nonsquare in $\gf(3^m)$. 
Hence, the equality of (\ref{eqn-X5}) cannot be possible. This contradiction proves the desired conclusion 
of this lemma.    
\end{proof}

In order to prove Theorem \ref{thm-HKMcodes}, we have to do more preparations. We now define for each 
$a \in \gf(3)$ and each $b \in \gf(3^m)^*$, 
$$ 
N_{(b,a)} =|\{x \in \gf(3^m): \tr(x+x^\ell)=0 \mbox{ and } \tr(bx)=a \}|. 
$$ 
One can easily prove that 
$$ 
|D \cup (-D)|=3^{m-1}-1. 
$$
It then follows that 
\begin{eqnarray}\label{eqn-Nba}
N_{(b,1)}=N_{(b,2)} \mbox{ and } N_{(b,0)}+N_{(b,1)}+N_{(b,2)}=3^{m-1}.      
\end{eqnarray} 

To determine the weight distribution of the code of Theorem \ref{thm-HKMcodes}, we need to find out $N_{(b,0)}$ 
for each $b \in \gf(3^m)$. Note that $\ell=e^2-e+1$ is odd and $\gcd(3^m-1, e+1)=2$. We know that $N_{(b,0)}$ 
is equal to the number of solutions $x \in \gf(3^m)$ of the following set of equations 
\begin{eqnarray}\label{eqn-HKM2}
\left\{ \begin{array}{l}
\tr(y^{e+1}+y^2)=0, \\
\tr(by^{e+1})=0.  
\end{array}
\right. 
\end{eqnarray}

We are now ready to prove the following lemma. 

\begin{lemma} \label{lem-July221}
For each $b \in \gf(3^m)^*$, $N_{(b,0)}$ has the following three possible values: 
$$ 
3^{m-2}, \  3^{m-2} \pm 2 \times 3^{2(h-1)}. 
$$
\end{lemma}

\begin{proof}
The number of solutions $y \in \gf(3^m)$ of (\ref{eqn-HKM2}) is equal to $N_{(b,0)}$ and is given by 
\begin{eqnarray}\label{eqn-8terms}
N_{(b,0)} &=& \frac{1}{9} \sum_{y \in \gf(3^m)} 
                        \left(\sum_{z_1 \in \gf(3)} \epsilon_3^{z_1\tr(y^{e+1} + y^2)} \right) 
                        \left(\sum_{z_2 \in \gf(3)} \epsilon_3^{z_2\tr(by^{e+1})} \right) \nonumber \\
&=&  3^{m-2} + \frac{1}{9} \left(\sum_{y \in \gf(3^m)}  \epsilon_3^{\tr(y^{e+1} + y^2)} + 
                         \sum_{y \in \gf(3^m)}  \epsilon_3^{-\tr(y^{e+1} + y^2)}  \right) + \nonumber \\   
& & \frac{1}{9} \left(\sum_{y \in \gf(3^m)}  \epsilon_3^{\tr(by^{e+1})} + 
                         \sum_{y \in \gf(3^m)}  \epsilon_3^{-\tr(by^{e+1})}  \right) + \nonumber \\ 
& & \frac{1}{9} \left(\sum_{y \in \gf(3^m)}  \epsilon_3^{\tr((1+b)y^{e+1} + y^2)} + 
                         \sum_{y \in \gf(3^m)}  \epsilon_3^{-\tr((1+b)y^{e+1} + y^2)}  \right) + \nonumber \\                          
& & \frac{1}{9} \left(\sum_{y \in \gf(3^m)}  \epsilon_3^{\tr((1-b)y^{e+1} + y^2)} + 
                         \sum_{y \in \gf(3^m)}  \epsilon_3^{-\tr((1-b)y^{e+1} + y^2)}  \right).                                             
\end{eqnarray} 

By Lemmas \ref{lem-rankQ1} and \ref{lem-ZhouDing}, we have 
$$ 
\frac{1}{9} \left(\sum_{y \in \gf(3^m)}  \epsilon_3^{\tr(y^{e+1} + y^2)} + 
                         \sum_{y \in \gf(3^m)}  \epsilon_3^{-\tr(y^{e+1} + y^2)}  \right)=0 
$$

It follows from Lemmas \ref{lem-LF081} and \ref{lem-ZhouDing} that  
$$ 
\frac{1}{9} \left(\sum_{y \in \gf(3^m)}  \epsilon_3^{\tr(by^{e+1})} + 
                         \sum_{y \in \gf(3^m)}  \epsilon_3^{-\tr(by^{e+1})}  \right)=0. 
$$

Combining Lemmas \ref{lem-DingCS} and \ref{lem-ZhouDing}, we know that at least one of 
the last two sums of the form $\frac{1}{9}(\ldots + \ldots)$ in (\ref{eqn-8terms}) is equal to $0$. 

The desired conclusion then follows from (\ref{eqn-8terms}) and Lemma \ref{lem-ZhouDing}. 
This completes the proof. 
\end{proof}

The following lemma follows from Lemma \ref{lem-July221} and (\ref{eqn-Nba}). 

\begin{lemma}\label{lem-DDD} 
The triple $(N_{(b,0)}, N_{(b,1)}, N_{(b,2)})$ takes on only the following three possible values: 
\begin{eqnarray*}
&& (3^{m-2}, \ 3^{m-2}, \ 3^{m-2}), \\  
&& (3^{m-2}+2\times 3^{2(h-1)}, \ 3^{m-2}- 3^{2(h-1)}, \ 3^{m-2}-3^{2(h-1)}), \\ 
&& (3^{m-2}-2\times 3^{2(h-1)}, \ 3^{m-2}+ 3^{2(h-1)}, \ 3^{m-2}+3^{2(h-1)}).  
\end{eqnarray*}
\end{lemma} 

The next lemma follows from Lemma \ref{lem-DDD} and the definition of $N_{(b,a)}$. 

\begin{lemma}\label{lem-chi1} 
For any $b \in \gf(3^m)^*$, $\chi_1(bD_0)$ takes on only one of the three values: 
$$ 
-1, \ 3^{2h-1}-1,   \ -3^{2h-1}-1,  
$$
where $D_0=D \cup (-D)$. 
\end{lemma}  

Finally, we are ready to prove Theorem \ref{thm-HKMcodes}. 

\subsection*{{\bf Proof of Theorem \ref{thm-HKMcodes}:}} 

For $x \in \gf(3^m)^*$, it follows from (\ref{eqn-weight}) and Lemma \ref{lem-chi1} that the codeword 
$\bc_x$ in (\ref{eqn-mcodeword}) has the following three weights: 
$$ 
w_1:=3^{3h-2}+3^{2h-2}, \ w_2:=3^{3h-2}, \ w_3:=3^{3h-2}-3^{2h-2}. 
$$
It is easy to prove that the dual code $\C_D^\perp$ has minimum weight at least $3$. 
We now determine the number $A_{w_i}$ of codewords with weight $w_i$ in $\C_{D}$. 
The first three Pless Power Moments \cite[p.260]{HP} lead to the following system of equations:  
\begin{eqnarray}\label{eqn-wtdsemibentfcode66}
\left\{ 
\begin{array}{lll}
A_{w_1}+A_{w_2}+A_{w_3} &=& 3^{3h}-1, \\
w_1A_{w_1}+w_2A_{w_2}+w_3A_{w_3} &=& 3^{3h-1}(3^{3h-1}-1), \\ 
w_1^2A_{w_1}+w_2^2A_{w_2}+w_3^2A_{w_3} &=& 3^{6h-3}(3^{3h-1}-1).  
\end{array}
\right. 
\end{eqnarray} 
Solving this set of three equations proves the weight distribution in Table \ref{tab-HKMcodes}.  

\begin{example} 
Let $h=1$. Then the code $\C_D$ of Theorem \ref{thm-HDMds} is a $[4, 3, 2]$ ternary code with weight 
enumerator $1+12z^{2}+8z^{3}+6z^{4}$ according to Magma, which confirms the result of 
Theorem \ref{thm-HKMcodes}. 
\end{example} 

\begin{example} 
Let $h=3$. Then the code $\C_D$ of Theorem \ref{thm-HDMds} is a $[3280, 9, 2106]$ ternary code with weight 
enumerator $1+756z^{2106}+18224z^{2187}+702z^{2268}$ according to Magma, which confirms the result of 
Theorem \ref{thm-HKMcodes}. 
\end{example}

\section{Concluding remarks} 

Although the idea of constructing linear codes in this paper is simple, one-weight codes, two-weight codes, 
and three-weight codes are constructed with those 2-designs. The codes are interesting, as 
one-weight codes, two-weight codes and three-weight codes have applications in secret sharing 
\cite{ADHK} and authentication codes \cite{CX05}. There is a survey on two-weight codes \cite{CK85}. 
Some interesting two-weight and three-weight codes were presented in \cite{CG84}, \cite{CW84}, 
\cite{Choi}, \cite{FL07}, \cite{LiYueLi}, \cite{LiYueLi2},  \cite{Xia}, and \cite{ZD13}. 

There are many other types of difference sets in $(\gf(p^m), \,+)$ and $(\gf(p^m)^*, \,\times)$ \cite{DingDScodes}. 
which give automatically linear codes within the framework of the construction of this paper. But it may be difficult to 
determine the parameters of these codes. The reader is cordially invited to attack this problem.   

Cyclic difference sets were employed to construct constant-weight codes in \cite{LZH}. A cyclic code approach 
to bent functions over $\gf(2)$ and $\Z_4$ is given in \cite{Wolf}. Almost perfect nonlinear functions and almost 
bent functions are employed to construct linear codes in \cite{CCZ}. Kerdock codes are also related to bent 
functions \cite{Carlet}. Some three-weight binary codes are also presented in \cite{MS77}[Theorems 33 and 34]. 
A related construction of linear codes is presented in \cite{MesagerDCC}.  In some of these references, a linear 
code over $\gf(q)$ is constructed with a highly nonlinear function $f$ from $\gf(q)$ to $\gf(q)$ and is defined by 
$$ 
\C(f)=\{ \bc=(\tr(af(x)+bx)_{x \in \gf(q)^*}: a \in \gf(q), \ b \in \gf(q)\}. 
$$
Its length is $q-1$, and its dimension is usually $2m$. The dual of $\C(f)$ has usually dimension $q-1-2m$. This 
approach gives a coding-theory characterisation of APN monomials, almost bent functions, and semibent functions (see, for 
examples, \cite{CCZ}, \cite{CCD00} and \cite{HX01}).  Theorem \ref{thm-BooleanCodes} of this papers gives 
automatically another coding-theory characterisation of APN monomials, almost bent functions, and semibent functions, 
though it looks simple and easy to derive.    

The construction of linear codes with 2-designs in this paper is different from all these constructions due to the 
difference in the dimension of the codes. The codes dealt with in this paper have dimension usually $m$, and 
length $n$, which is smaller than $q-1$ and may not divide $q-1$. 
 
It is obvious that the construction of linear codes 
of this paper is different from the classical one employing the incidence matrix of a design due to the difference in the length 
of the codes. 

Any linear code over $\gf(p)$ can be employed to construct secret sharing schemes \cite{ADHK,CDY05,YD06}. In order to 
obtain secret sharing schemes with interesting access structures, we would like to have linear codes $\C$ such that 
$w_{\min}/w_{\max} > \frac{p-1}{p}$ \cite{YD06}, where $w_{\min}$ and $w_{\max}$ denote the minimum and maximum 
nonzero weight of the linear code. 

The one-weight code of Theorem \ref{thm-part1} can be employed to construct secret sharing schemes using the approach 
of \cite{ADHK}. For the two-weight and three-weight codes over $\gf(p)$ obtained in this paper, we have 
$
\frac{w_{\min}}{w_{\max}} > \frac{p-1}{p}, 
$ 
provided that $m$ is large enough. Therefore, almost all the codes of this paper can be employed to construct secret sharing schemes 
with certain interesting access structures \cite{YD06}.

\section*{Acknowledgements}

The author would like to thank the reviewers and the Associate Editor, Dr. Sihem Mesnager, for their constructive comments that much 
improved the presentation and quality of this paper.

\end{document}